\documentclass{article}

\usepackage{natbib}
\usepackage{bibentry}
\usepackage{amssymb}
\usepackage{amsmath,amsthm,amsfonts}
\usepackage{thmtools,bm}
\usepackage{mathtools,algorithm}
\usepackage{jinshuomacros}
\usepackage[noend]{algpseudocode}
\usepackage{algorithmicx}
\usepackage{mathrsfs}
\usepackage{tikz-cd}
\usepackage{geometry}
\usepackage{graphicx}

\newif\ifhideproofs

\ifhideproofs
\usepackage{environ}
\NewEnviron{hide}{}

\fi

\iftrue
\declaretheoremstyle[bodyfont=\normalfont]{normalbody}
\newtheorem{theorem}{Theorem}[]
\newtheorem{lemma}[theorem]{Lemma}

\newtheorem{corollary}[theorem]{Corollary}
\newtheorem{proposition}[theorem]{Proposition}
\newtheorem{fact}[theorem]{Fact}

\newtheorem{definition}[theorem]{Definition}
\fi

\usepackage{amssymb}
\usepackage{tcolorbox}

\title{Classification Protocols with Minimal Disclosure}

\author{Jinshuo Dong\thanks{Northwestern University. Email: \texttt{jinshuo@northwestern.edu}.} \and Jason Hartline\thanks{Northwestern University. Email: \texttt{hartline@northwestern.edu}.} \and Aravindan Vijayaraghavan\thanks{Northwestern University. Email: \texttt{aravindv@northwestern.edu}.}}

\usepackage{hyperref}
\usepackage{cleveref}
\newcommand{\cH}{\mathcal{H}}

\newcommand{\Leak}{\mathrm{Leak}}
\newcommand{\CP}{{\mathcal C}^*}

\newcommand{\cone}{\mathrm{cone}}

\newcommand{\Dir}{{\mathcal D}}
\newcommand{\SH}{{\mathcal{H}}}
\newcommand{\Safe}{\text{Safe}}
\newcommand{\Verts}{\text{Verts}}

\newcommand{\Sm}{S_-}
\newcommand{\Sp}{S_+}
\newcommand{\Am}{A_-}
\newcommand{\Ap}{A_+}

\newcommand{\xm}{x_{-}}
\newcommand{\xp}{x_{+}}

\newcommand{\proto}{{\mathcal M}}
\newcommand{\protor}{\proto^{R}}
\newcommand{\strat}{\sigma}

\newcommand{\reals}{{\mathbb R}}

\newcommand{\TTP}{Trent} 
\newcommand{\cTTP}{Trent} 

\newcommand{\nadded}[1]{#1}
\newcommand{\nsubtracted}[1]{}

\newcommand{\added}[1]{#1}

\newcommand{\iprod}[1]{\langle #1 \rangle}
\usepackage{algpseudocode}

\begin{document}

\date{}

\maketitle
\begin{abstract}
  We consider multi-party protocols for classification that are
  motivated by applications such as e-discovery in court proceedings.
  We identify a protocol that guarantees that the requesting party
  receives all responsive documents and the sending party discloses
  the minimal amount of non-responsive documents necessary to prove
  that all responsive documents have been received.  This protocol can
  be embedded in a machine learning framework that enables automated
  labeling of points and the resulting multi-party protocol is
  equivalent to the standard one-party classification problem (if the
  one-party classification problem satisfies a natural
  independence-of-irrelevant-alternatives property).  Our formal
  guarantees focus on the case where there is a linear classifier that
  correctly partitions the documents.
\end{abstract}

\section{Introduction}

This paper considers the multi-party classification problem that
arises in document review for discovery in legal proceedings.  The
plaintiff (henceforth: Bob) issues a {\em request for production} to
the defendant (henceforth: Alice).  The legal team of Alice is then
accountable for reviewing all documents and provides the responsive
ones.  \citet{GC-10} show this manual process can be significantly
improved by automation.  A potential issue with the adoption of this
technology, however, is that automation could reduce transparency and
accountability, and the accuracy and completeness of this process
relies critically on the accountability of Alice's legal team and its
obligations under the rules of professional responsibility.

In addition to accountability, \citet{GK-15} identify significant
problems with the above method for discovery. First, the defendant (Alice)
bears most of the cost of reviewing and selecting the responsive
documents and this asymmetry could lead the plaintiff (Bob) to exploit
such costly requests.  Second, it misaligns the incentives of the
Alice's legal team and causes the team, on the grounds of professional
responsibility, to conduct work to benefit its adversary.

Another possible way to implement requests for production places the
effort and accountability on the plaintiff.  Bob issues a request for
production to Alice.  Alice delivers all the documents to Bob's legal
team.  Bob's legal team identifies the responsive documents (and
discards the non-responsive ones).  Of course, there is now a risk
that Bob's legal team might learn facts from the documents not
specified in the request for production.  Alice and Bob may enter into a
confidentiality agreement under the order of the court to protect the
disclosure of Alice's private information and Bob's legal team should
operate under its obligations under the rules of professional
responsibility.

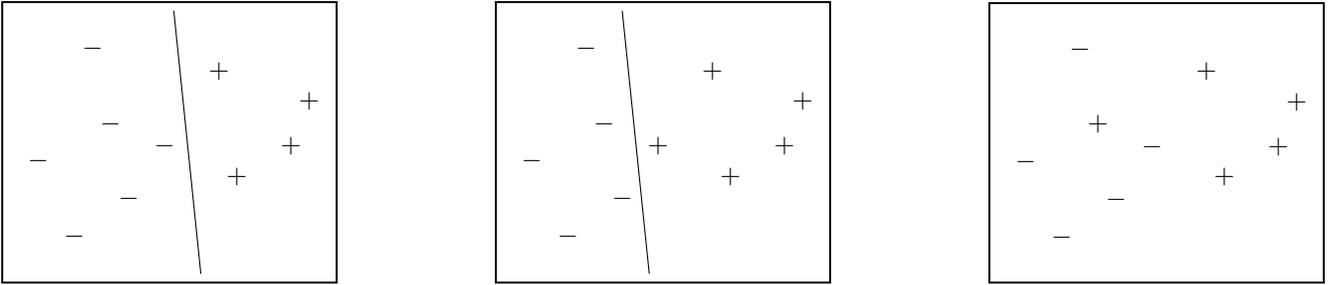
\begin{figure*}[tb]
\fbox{
\begin{tikzpicture}[xscale=1.2]
\node at (1,2) {$-$};
\node at (2,1.5) {$-$};
\node at (1.4,1) {$-$};
\node at (1.6,3.5) {$-$};
\node at (2.4,2.2) {$-$};
\node at (1.8,2.5) {$-$};

\node at (3,3.2) {$+$};
\node at (3.2,1.8) {$+$};
\node at (4,2.8) {$+$};
\node at (3.8,2.2) {$+$};

\draw [-] (2.8,0.5) -- (2.5,4);
\end{tikzpicture}}
\hfill
\fbox{
\begin{tikzpicture}[xscale=1.2]
\node at (1,2) {$-$};
\node at (2,1.5) {$-$};
\node at (1.4,1) {$-$};
\node at (1.6,3.5) {$-$};
\node at (2.4,2.2) {$+$};
\node at (1.8,2.5) {$-$};

\node at (3,3.2) {$+$};
\node at (3.2,1.8) {$+$};
\node at (4,2.8) {$+$};
\node at (3.8,2.2) {$+$};

\draw [-] (2.3,0.5) -- (2,4);
\end{tikzpicture}}
\hfill
\fbox{
\begin{tikzpicture}[xscale=1.2]
\node at (1,2) {$-$};
\node at (2,1.5) {$-$};
\node at (1.4,1) {$-$};
\node at (1.6,3.5) {$-$};
\node at (2.4,2.2) {$-$};
\node at (1.8,2.5) {$+$};

\node at (3,3.2) {$+$};
\node at (3.2,1.8) {$+$};
\node at (4,2.8) {$+$};
\node at (3.8,2.2) {$+$};

\path (2.3,0.5) -- (2,4);
\end{tikzpicture}}
\caption{Left: The labeled points are linearly separable; Middle: The labeled points are linearly separable when the right-most negative point is relabeled as positive; Right: The labeled points are not linearly separable when the center negative point is relabeled as positive.  The right-most negative point is critical; the center negative point is non-critical.}
\label{f:critical}
\end{figure*}

This paper aims to understand multi-party binary classification
protocols that rely as little as possible on external means of
accountability.  We
aim for protocols that satisfy three main properties:
\begin{enumerate}
    \item (Correct) Bob receives all responsive documents.
    \item (Minimal) Alice minimizes privacy loss (as few non-responsive documents as possible are revealed to Bob).
    \item (Computationally Efficient): Algorithms run by all parties are computationally efficient.
\end{enumerate}
We will also be interested in a fourth property which our protocol will satisfy:
\begin{enumerate}
\item[(4)] (Truthful) Alice's best strategy in the protocol is to truthfully reveal the set of relevant documents.
\end{enumerate}

Our protocols will make use of a trusted third party, Trent.  Trusted
third parties are common in the design of protocols, and can often be
replaced with secure multi-party computation \citep{yao-86,GMW-87}.
For e-discovery of electronic mail, many companies already use third
parties for email storage, and perhaps these third parties can take on
the role of the trusted third party in our protocol.  Alternatively,
the court system could provide the trusted third party.

After introducing a trusted third party, there is now a third possible
protocol: Bob communicates a classifier to Trent and Trent uses the
classifier identify the responsive documents, checks these documents
with Alice, and then communicates them to Bob.  The problem with this
approach is that, while we assume Bob's legal team can identify
whether or not any given document is responsive, we do not assume that
Bob can succinctly communicate such a labeling strategy in the form of
a machine executable classifier.  For Bob to produce such a classifier
Bob needs real documents that only Alice possesses.  We furthermore do
not assume that Trent is capable of non-mechanical tasks.

Our protocol is based on two assumptions on the environment.  First,
we assume that Alice can and will provide Trent with all documents
before the protocol begins.  We view this assumption as much weaker
than the current standard assumption of the leading paragraph where
Alice is required to provide only the relevant documents.  Providing
all documents is a weaker requirement than providing the relevant
documents because there is no potential discretion involved.  Alice's
legal team cannot claim to think an document was not responsive as a
justification for not providing it.  Our second assumption is that both
Alice and Bob can determine responsiveness of documents and if there is a
disagreement in responsiveness, that this disagreement can be resolved
by the court.

Our problem is one of multi-party classification, dividing the
documents (henceforth: points) into responsive (positive points) and
non-responsive (negative points).  We assume that there is a
classifier that is consistent with the labeling of the documents
(See \Cref{f:critical}).  A key construct in the protocol is, given a
set of alleged positive points and a set of alleged negative points
that are separable, identifying all the other points that could be labeled
as positive by a consistent classifier.  We will refer to these points as {\em leaked}.  A key quantity for
our protocol is the {\em critical points}: negative points that are
leaked when all other points are known.  It is easy to see that there
is no way Trent can be convinced that a critical point is negative
without confirming its negative label with Bob.

\vspace{40pt}

With these constructs we define the {\em critical points protocol}:
\begin{enumerate}
\item[(0)] Alice discloses all points to Trent.
\item Alice discloses to Trent which of the points she alleges as positive.
\item Trent assumes that all remaining points are negative and
  computes the alleged critical points.
\item Trent sends the alleged positive and alleged critical points to Bob.
\item Bob labels these points and sends the labels to Trent.
\item Trent checks Bob's and Alice's labels agree (resolving any disagreement in court).
\item Trent sends the leaked points corresponding to the correctly labeled points to Bob. 
\end{enumerate}

Our main protocol and results will be for binary classification with linear
classifiers (See \Cref{f:critical}).  In these settings there is a projection from document
space into a high-dimensional space of real numbers.  Classifiers are
given by hyperplanes that partition the space into two parts, the
positives and the negatives.  The assumption that there exists such a
consistent classifier is known as the {\em realizable} setting  or the {\em (linearly) separable} setting. 
Our main result
is that for linear classification in the realizable setting, the critical
points protocol is {\em correct}, {\em minimal}, and {\em computationally efficient}. (See \Cref{sec:ccp}
for formal definitions and theorem statements.)

While we focus on linear classification for exposition, our main result also extends to more powerful kernel-based classifiers like kernel support vector machines. Kernel methods embed the input space into a feature space that is higher-dimensional (potentially infinite dimensional), where the data is potentially linearly separable. Hence they capture more expressive hypothesis classes like polynomial threshold functions\footnote{The label of $x$ is given by the sign of a polynomial $p(x)$; linear classifiers correspond to the special case of degree-$1$ polynomials.}, and even neural networks in some settings~\citep{kernelbook,ShaiBSSbook, jacotetal}. See \Cref{sec:kernel} for details.

We can also show that the basic critical points protocol can be
embedded within a machine learning framework that includes several of
the technology-assisted review processes studied by \citet{CG-14}.  In
this framework, there is a large universe of documents.  This large
universe of documents is sampled.  The critical points protocol is run
on the sample with labels provided by hand by the legal teams of Bob
and Alice as specified.  When the protocol terminates with Bob
possessing both the critical negative points and the positive points,
Bob selects a classifier that is consistent with these points.  Bob
reports this classifier to Trent who checks that it is consistent with the labeled points and then applies it to the universe
of points and gives Bob all the points that are classified as
positive.  We prove that if Bob's classification algorithm satisfies
{\em independence of irrelevant alternatives}, i.e., if the classifier
selected is only a function of the set of consistent classifiers for
the labeled points, then the outcome of this process is equivalent to
the outcome of an analogous single-party classification procedure. 
In the case of linear classification and kernel-based classifiers, this can be instantiated with the support vector machine (SVM) algorithm that we prove satisfies the IIA property. 
(See \Cref{sec:mlframework} for formal statements.)

\paragraph{Related Work.}
Our work contributes to a growing literature on the theory of machine
learning for social contexts.  In this literature it is not enough for
the algorithm to have good performance in terms of error, but it must
also satisfy key definitions to be usable.  Like a number of problems
in this space, the gold-standard result is a reduction from the learning
problem with societal concerns to the learning problem without such
concerns.  For example, \citet{DHPR-12} construct fair classifiers
from non-fair classifiers.  A key perspective of this approach is it
enables the machine learning algorithm designers to plug in their
favorite algorithms, but results in a system with the desired societal
properties, in their case, fairness. Our results for the machine learning framework
in \Cref{sec:mlframework} are of a similar flavor: our protocol can be used in conjunction with 
any learning algorithm satisfying the IIA property to extend it to the multi-party setting.  

\citet{GRSY-21} consider interactive protocols for PAC (provably
approximately correct) learning.  They ask whether a verifier can be
convinced that a classifier is approximately correct with far fewer
labeled data points than it takes to identify a correct classifier.  In
the realizable case, the answer is yes.  More generally, they show
that there are classification problems where it is significantly
cheaper in terms of labeled data points; and there are classification
problems where it is no cheaper.  Connecting to our model, their
prover corresponds to Alice, their verifier corresponds to Bob.  The
big difference between their model and ours is that they assume that
Bob (the verifier) can freely sample labeled data points.  In our
model Bob does not have access to the data without getting it from
Alice. Moreover the main challenge in our setting is for Alice and Trent to convince Bob that no relevant
documents were left out. 















\section{Critical Points Protocol: Definitions, Protocol and Guarantees}
\label{sec:ccp}
\label{sec:cpp}

There are three parties: Alice (defendant), Bob (plaintiff) and {\TTP}
who is a trusted third party.  Alice has a set of data points
$S \subset \R^n$ (potentially the training samples), that is comprised
of positive examples $\Sp$ and negative examples $\Sm$ with their
disjoint union denoted by $\Sp \sqcup \Sm = S$. Alternately, each data point
corresponds to a labeled example of the form $(x,y)$ where
$x \in \R^n$ and $y \in \{\pm 1\}$, where $y=+1$ if $x \in \Sp$ and
$y=-1$ if $x \in \Sm$. These labeled examples are assumed to be
(strictly) linearly separable i.e., there exists $d \in \reals^n$ and
$c \in \R$ such that $y=h(x)=\text{sign}(d \cdot x - c)$.

There are potentially several rounds of interaction between Alice, Bob
and {\TTP}. We will adhere to the revelation principle (\Cref{s:truth})
and restrict attention to protocols where Alice only interacts once
and is asked to specify the positive labels $\Sp$ (the other labels
$S \setminus \Sp$ are assumed to be the negative points $\Sm$).  For a
truthful mechanism $\proto$, we will denote by $\proto(\Ap,\Sp) \subset
S$ the set of points that are revealed to Bob eventually when Alice
reports $\Ap$ and the true labels are $\Sp \subset S$. Here
$\proto(\Ap,\Sp)$ is the output of the protocol.  We aim for a
protocol which satisfies the following four properties.

\begin{definition}
Protocol $\proto$ properties on all data sets $S = \Sp \sqcup \Sm$ and all
reports $\Ap \subset S$:
\begin{enumerate}
\item (Correct) The positive points are revealed to Bob, \\
i.e., $\proto(\Ap,\Sp) \supseteq \Sp.$
\item (Minimal) (If Alice reports truthfully) the protocol minimizes the number of negative points revealed, i.e., \\
$|\proto(\Sp,\Sp) \setminus \Sp|$ is minimized.
\item (Computational Efficiency) The algorithms run by all parties are computationally tractable.
\item (Truthful) Alice's best strategy is to truthfully reveal the set of relevant documents, i.e.,  $\Ap = \Sp$ minimizes $\proto(\Ap,\Sp)$.
\end{enumerate}
\end{definition}

In the protocol we will define, Alice and Bob will be expected simply
to label points.  The complex computations will be mechanically
performed by {\TTP}, the trusted third party.  The basic computation
performed by {\TTP} is the $\Leak$ operator which, given a subset of
linearly separable points labeled as positives and negatives
$\Ap \sqcup \Am = A$, determines the set of all the points $S$ that
are labeled as positive by some classifier consistent with the labels
of $A$. Let $\cH=\{h(x) = \text{sign}(d \cdot x - c): d \in \reals^n,
c \in \reals \}$ denote the set of all linear classifiers over
$\reals^n$.

\begin{definition}
The consistent classifiers for points $\Ap \sqcup \Am = A$ is $\SH(\Ap,\Am)=$ $$ \{h \in \cH : \forall \xp \in \Ap,\ h(\xp) = +1 \text{ and } \forall \xm \in \Am,\ h(\xm) = -1\}.$$
\end{definition}

\begin{definition} The leak operator $\Leak(\Ap,\Am)$ is all the points in $S$ are classified by positive by some consistent classifier.
$$\Leak(\Ap,\Am) = \{x \in S : \exists h \in \SH(\Ap,\Am),\ h(x) = +1\} \added{\cup \Am}$$
\end{definition}

\begin{lemma} For linear classification, whether or not a point $x \in S$ is in $\Leak(\Ap,\Am)$ is a linear classification problem (and can be computed in polynomial time).
\end{lemma}

\begin{proof} We first check if $x \in \Am$; if yes, $x \in \Leak(\Ap, \Am)$. If not, the goal is to check if there is a linear classifier $h \in \cH$ that assigns a label $+1$ to all points in $\Ap \cup \{x\}$, and assigns the label $-1$ to all points in $\Am$. This is clearly a linear classification problem. This can be solved in polynomial time using the SVM algorithm (see Fact~\ref{fact:svm1} in \Cref{sec:svm}) or using a linear program.  
\end{proof}

In the introduction the critical points were defined as the points
that could are ambiguous with respect to the consistent classifiers
when Alice reports positives as $\Ap$ and all other points not labeled as positive by Alice are negatives.
The definition of critical points requires that $\Ap$ is linearly
separable from $S \setminus \Ap$.

\begin{definition} 
The critical points for set $\Ap \subset S$ are
$$
\CP(\Ap) = \{x \in S: x \in \Leak(\Ap,S \setminus \Ap \setminus \{x\})\}~ \added{\big\backslash A_+}.
$$
\end{definition}

We refer to the following algorithm as the critical points protocol
because, as we will subsequently prove (in \Cref{thm:verts}), 
\begin{equation} \label{eq:main:verts}
    \Leak(\Ap,\CP(\Ap))
= \Ap \cup \CP(\Ap).
\end{equation}  
Thus, if Alice truthfully reports $\Ap = \Sp$
then the protocol terminates with the only negative points disclosed being
$\CP(\Ap)$.

\begin{algorithm}
		\caption{Critical Points Protocol
		(CPP)}\label{alg:ccp}
                \begin{algorithmic}[1]
                \State Alice sends all points $S$ to {\TTP}.
		\State Alice sends alleged positive points $\Ap \subset S$ to {\TTP}.
                \If{$\Ap$ and $S \setminus \Ap$ are not separable}
                \State {\TTP} sends $S$ to Bob and the protocol ends.
                \EndIf
		\State {\TTP} computes critical points $\CP(\Ap)$ and sends $\Ap \cup \CP(\Ap)$ to Bob.
                \State Bob labels the points and sends labels to {\TTP}.
                \State {\TTP} checks that Bob and Alice's labels are consistent sending any disputed labels to be resolved by the court.  Denote the resulting labeled points by $\Ap' \sqcup \Am'$
                \State {\TTP} sends $\Leak(\Ap',\Am')$ to Bob and the protocol ends.
          	\end{algorithmic}
		\end{algorithm}




\begin{theorem}\label{thm:min}
  When the data points $S$ with positive samples $\Sp$ and negative
  samples $\Sm$ are linearly separable, the critical points protocol
  (\Cref{alg:ccp}) is (1) Correct, (2) Minimal i.e., for all $\Ap$
  $\proto(\Ap, \Sp) \supseteq \Sp \cup \CP(\Sp)$, and if Alice is
  truthful and reports $\Sp$ to {\TTP}, then $\proto(\Sp, \Sp) = \Sp
  \cup \CP(\Sp)$. Furthermore this protocol is (3) Computationally
  efficient and (4) Truthful.
\end{theorem}
We remark that minimality of the CPP protocol holds in a stronger
sense: in every {\em correct} protocol, Bob observes $\Sp \cup
\CP(\Sp)$; see \Cref{prop:minimal} for a proof. Please also see
\Cref{s:truth} for a revelation principle.

We now proceed to the proof of \Cref{thm:min}.  Note that Bob either
sees all of $S$, or he sees $\Leak(A'_+, A'_-)$
for some appropriate sets $\Ap' \subset \Sp$ and $\Am' \subset \Sm$.

\begin{proof}[Proof of \Cref{thm:min}]
The proof of the theorem follows from the following three
claims for any $\Ap \sqcup \Am \subset S$:
\begin{enumerate}
\item $\Leak(\Ap,\Am)$ is non-increasing in $\Ap$.

  For fixed $\Am$, and consider $\Ap' \supset \Ap$ that is separable
  from $\Am$.  Separability of $\Ap' \sqcup \Am$ implies that the new
  points $\Ap' \setminus \Ap$ were also previously leaked in
  $\Leak(\Ap,\Am)$.  Moreover, for other points in $S$, there are now
  more constraints on separating hyperplanes so only fewer of them
  will be leaked.  In total, no more points are leaked by
  $\Leak(\Ap',\Am)$.
    
\item $\Leak(\Ap,\Am)$ always contains $\CP(\Ap)$.

  Fix $\Ap$, whether or not point $x$ is leaked is monotone decreasing
  in $\Am \not\ni \x$.  Adding points to $\Am$ only adds constraints
  on separating hyperplanes making it only harder for $x$ to be
  leaked.  Since $x \in \CP(\Ap)$ is leaked by $\Leak(\Ap,S \setminus
  \Ap \setminus \{x\})$ then by monotonicity $x$ is leaked by
  $\Leak(\Ap,\Am)$ on all $\Am$ separable from $\Ap$.
  
\item Nothing additional is leaked on $\Am = \CP(\Ap)$, \\
i.e.,
  $\Leak(\Ap,\CP(\Ap)) = \Ap \cup \CP(\Ap)$.
  This claim will be argued separately by \Cref{thm:verts}, below.
  
\end{enumerate}
These claims combine to give the theorem as follows.  In the protocol,
$\Ap' \subset \Sp$ and $\Am' \subset \Sm$.  By the first claim we have\\
$\Leak(\Ap',\Am') \supset \Leak(\Sp,\Am')$.  However, $\Leak(\Sp,\Am')
\supset \Sm \cup \Am'$ and by the second claim $\Leak(\Sp,\Am')
\supset \CP(\Sp)$; \\ 
thus, 
$\Leak(\Sp,\Am') \supset \Sm \cup \Am' \cup
\CP(\Sp)$ which, of course, is a superset of $\Sp \cup \CP(\Sp)$ which
is equal to $\Leak(\Sp,\CP(\Sp))$ by the third claim.  This latter
minimal outcome is obtained by truthtelling.
\end{proof}

\begin{figure*}[tb]
\fbox{
\begin{tikzpicture}[xscale=.66]

\draw [color=lightgray,ultra thick] (1,0) -- (5,4);
\draw [color=lightgray,ultra thick] (1,4) -- (5,0);
\draw [color=lightgray,ultra thick] (8,0) -- (0,4);
\draw [color=lightgray,ultra thick] (0,0) -- (8,4);
\draw [color=lightgray,ultra thick] (6,0) -- (6,4);

\node at (1,1) {$-$};
\node at (1,2) {$-$};
\node at (1,3) {$-$};
\node at (2,1) {$-$};
\node at (2,2) {$-$};
\node at (2,3) {$-$};
\node at (3,2) {$-$};

\node at (6,1) {$+$};
\node at (6,3) {$+$};
\node at (7,1) {$+$};
\node at (7,3) {$+$};

\path (8,0) -- (8,4);
\path (0,0) -- (0,4);
\end{tikzpicture}}
\hfill
\fbox{
\begin{tikzpicture}[xscale=.66]

\fill [color=lightgray] (0,0) -- (2,1) -- (3,2) -- (2,3) -- (0,4) -- (8,4) -- (6,3) -- (6,1) -- (8,0);

\node at (1,1) {$-$};
\node at (1,2) {$-$};
\node at (1,3) {$-$};
\node at (2,1) {$-$};
\node at (2,2) {$-$};
\node at (2,3) {$-$};
\node at (3,2) {$-$};

\node at (6,1) {$+$};
\node at (6,3) {$+$};
\node at (7,1) {$+$};
\node at (7,3) {$+$};

\path (8,0) -- (8,4);
\path (0,0) -- (0,4);
\end{tikzpicture}}
\hfill
\fbox{
\begin{tikzpicture}[xscale=.66]

\fill [color=lightgray] (0,0) -- (2,1) -- (3,2) -- (2,3) -- (0,4);

\node at (1,1) {$-$};
\node at (1,2) {$-$};
\node at (1,3) {$-$};
\node [circle,draw] at (2,1) {$-$};
\node at (2,2) {$-$};
\node [circle,draw] at (2,3) {$-$};
\node [circle,draw] at (3,2) {$-$};

\node [circle,draw] at (6,1) {$+$};
\node [circle,draw] at (6,3) {$+$};
\node [circle,draw] at (7,1) {$+$};
\node [circle,draw] at (7,3) {$+$};

\path (8,0) -- (8,4);
\path (0,0) -- (0,4);
\end{tikzpicture}}
\hfill
\caption{Left: Extreme separating hyperplanes are depicted.  Center: all separating hyperplanes are depicted.  Right: the space $\Safe(\Ap,\Am)$ is depicted where $\Ap = \Sp$ and $\Am = \CP(\Sp)$ are the circled minuses.  Note that $\Verts(\Safe(\Ap,\Am)) = \CP(\Sp)$.  All circled points are disclosed by the protocol when Alice reports $\Ap = \Sp$.}
\label{f:safe}
\end{figure*}
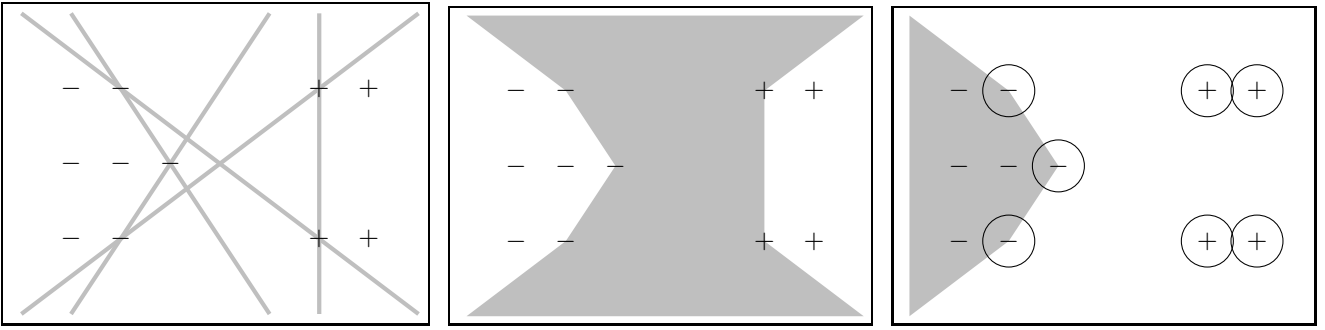

In the remainder of this section we prove equation \eqref{eq:main:verts} as
\Cref{thm:verts}.

\begin{theorem} \label{thm:verts}
For any $\Ap \subset S$, $\Leak(\Ap,\CP(\Ap)) = \Ap \cup \CP(\Ap).$
\end{theorem}

We have already argued that $\Leak(\Ap,\CP(\Ap)) \supset \Ap \cup
\CP(\Ap)$, it suffices to show that no other points in $x \in S
\setminus \Ap \setminus \CP(\Ap)$ are in $\Leak(\Ap,\CP(\Ap))$.  To do
so, we identify from $\Ap$ a polyhedron $S'$ that contains all of $S
\setminus \Ap$ and show that (a) its vertices $V'$ are exactly
$\CP(\Ap)$ and (b) $\Leak(\Ap,V') = \Ap + V'$.  Specifically, when the
vertices of this polyhedron are disclosed, the other negative points, which are
all within the polyhedron, are not leaked.

The polyhedron, $\Safe(\Ap,\Am)$ is defined as follows.  Denote the
maximum in direction $d \in \reals^n$ in a set $X$ by $\max_d(X) =
\max_{x \in X} d \cdot x$ and, respectively, the minimum by
$\min_d(X)$.
Denote the (weakly) separating directions (See \Cref{f:safe}) for
linearly separable points $\Ap$ and $\Am$ by $\Dir(\Ap,\Am)$ defined
as
\begin{equation}\label{eq:dir}
    \Dir(\Ap,\Am) = \{d \in \reals^n : \max\nolimits_d(\Am) \leq \min\nolimits_d(\Ap)\}.
\end{equation}
Define the convex subspace of points (a polyhedron) that would not be leaked by disclosing $\Ap \sqcup \Am$ as {\em safe points} (see \Cref{f:safe}) and denote this subspace by
\begin{equation}\label{eq:safe}
\Safe(\Ap,\Am) = \{x \in \reals^n : \forall d \in \Dir(\Ap,\Am), d \cdot x \leq \max\nolimits_d(\Am)\}.
\end{equation}
I.e., a point is safe if in all separating directions there is a
negative point disclosed in $\Am$ that is at least as big (in this direction).

In what follows $\Verts(X)$ denotes the vertices of convex set $X$;
these are the points that are unique maximizers in any direction.  The
following lemma extends the fundamental theorem of linear programming
to polyhedra like $\Safe(\Ap,\Am)$ that are defined by the separating
hyperplanes between finite point sets.  It shows that, if it is finite,
the optimal point in some direction is attained at a vertex.  Recall
that the definition of separable for $\Ap \sqcup \Am$ is strict.  This
strictness is important.  For example if all points in $\Ap$ and $\Am$
lie on the same hyperplane then $\Ap$ is weakly on one side and $\Am$
is weakly on the other, in a sense, they are weakly separated.  In
this case, $\Safe(\Ap,\Am)$ is a halfspace and halfspaces have no
vertices.  Moreover, if $\Sp \sqcup \Sm$ are on the same hyperplane
then there are no critical points $\CP(\Sp) = \emptyset$.  Thus, the
strictness of separation will play an important role in the proof of
the main theorem.

\begin{lemma}
  \label{f:vertex} For separable $\Ap \sqcup \Am$ and $S'
  = \Safe(\Ap,\Am)$, if direction $d\in \reals^n$ has a finite
  optimizer in $S'$ then $d$ is optimized in $S'$ at a vertex
  $v \in \Verts(S')$.
\end{lemma}

\begin{proof}
The proof follows because $\Safe(\Ap,\Am)$ does not contain a line (of
infinite length) and the fundamental theorem of linear programming which states that for
polyhedra that do not contain any line, every direction with a finite
optimizer is optimized at a vertex~\citep[see Theorems 2.6 and 2.7 in][]{bertsimas-LPbook}.

Suppose the convex hulls of $\Ap$ and $\Am$ are strictly separated at distance $2\epsilon >
0$ apart and consider a hyperplane $h$ that separates them with
distance $\epsilon$ from each.  Because $\Ap$ and $\Am$ are both
finite, their convex hulls are bounded and it is possible to rotate
$h$ small amount in any direction while still separating $\Ap$ and
$\Am$.  A line intersects a hyperplane if and only if the line and
hyperplane are \added{not} parallel.  For any line, one of the small rotations of
$h$ is not parallel and, thus, is intersecting.  This line,
therefore, does not lie completely within $\Safe(\Ap,\Am)$.
\end{proof}

We now observe that $\Dir$ and $\Safe$ behave the same way on second
parameter given by any of $\Am$, $S' = \Safe(\Ap,\Am)$, and $\Verts(S')$.

\begin{lemma} Let $S' = \Safe(\Ap,\Am)$, then
  \label{l:same} \begin{itemize} \item $\Dir(\Ap,\Am) = \Dir(\Ap,S')
  = \Dir(\Ap,\Verts(S'))$, and

  \item 
  $\max_d(\Am) = \max_d(S') = \max_d(\Verts(S'))$
    for all directions $d \in \Dir(\Ap,\Am)$.
  \end{itemize}
\end{lemma}

\begin{proof}
  For the first bullet: For the first equality, all points in $\Am$
  are safe so $S' \supset \Am$.  $\Dir(\Ap,\cdot)$ is bigger when
  its second argument is smaller. But every $d \in \Dir(\Ap,\Am)$ is
  also in $\Dir(\Ap,S')$ as the only points we add in $S'$ are smaller
  than the largest point in $\Ap$ in direction $d$.  For the second
  equality, \Cref{f:vertex} implies $\Dir(\Ap,S')
  = \Dir(\Ap,\Verts(S'))$.

  For the second bullet: By the definition of safe, $\max_d(\Am) =
  \max_d(S')$.  Specifically, all points we add are worse than points
  in $\Am$ in all relevant directions, but points in $\Am$ are also
  contained in $S'$ so the maximum values over these sets in direction
  $d$ direction are equal. \Cref{f:vertex} implies $\max_d(S') =
  \max_d(\Verts(S'))$.
\end{proof}

\begin{corollary}
  \label{c:same}
  Let $S' = \Safe(\Ap,\Am)$, then 
  $$\Safe(\Ap,\Am) = \Safe(\Ap,S') =
  \Safe(\Ap,\Verts(S')).$$
\end{corollary}

\begin{proof}
  By \Cref{l:same} all the terms in the definition of $\Safe$ that depend on the second parameter are the same.
\end{proof}

We conclude that if the vertices of $\Safe$ are disclosed, then no other negative points are leaked.

\begin{lemma}
\label{c:leak-verts}
For $S' = \Safe(\Ap,S \setminus \Ap)$ and $V' = \Verts(S')$, then 
$$
\Leak(\Ap,V') = \Ap \cup V'. 
$$
\end{lemma}

\begin{proof}
  $\Ap \cup V'$ are leaked by definition.  By \Cref{c:same}, $\Safe(\Ap,V') = \Safe(\Ap,S \setminus \Ap)$.  By the definition of $\Safe$, no other points in $\Safe(\Ap,S \setminus \Ap) \supset S \setminus \Ap$ are leaked.  
\end{proof}

\nadded{

With a view towards characterizing the vertices of Safe, the following lemma shows that all directions with finite maximizers within $\Safe$ are contained in $\Dir$. In what follows $\cone(X)$ denotes the set of points obtained by taking non-negative linear combinations of points in $X$. A set $X$ is a convex cone if and only if $\cone(X)=X$; it is said to be finitely generated if there exists a finite set of points $v_1, \dots, v_m$ such that $X=\cone(\{v_1, \dots, v_m\})$.  

\begin{lemma}\label{l:finitemax}
  For any $\Ap, \Am$, $\Dir(\Ap, \Am)$ is a convex cone. Moreover for any direction $d \in \R^n$,  $\max_{x \in \Safe(\Ap, \Am)} d \cdot x$ is finite if and only if $d \in \Dir(\Ap, \Am)$. 
\end{lemma}
\begin{proof}
  We first observe that $\Dir(\Ap, \Am)$ defined in \eqref{eq:dir} can equivalently be described as 
  \begin{equation} \label{eq:cone:desc}
      \Dir(\Ap, \Am) = \big\{ d: \forall a_- \in \Am, a_+ \in \Ap,~ d \cdot (a_- - a_+) \le 0 \big\},
  \end{equation}
  which corresponds to the solution set of a system of homogenous linear inequalities. Hence $\Dir(\Ap, \Am)$ is a convex cone since it is closed under non-negative combinations. 

  We now prove the second part. One direction is easy: if $ d \in \Dir(\Ap, \Am)$, then by the definition of $\Safe$, we have $\max_{x \in \Safe(\Ap, \Am)} \le \max_d(\Am)$ which is bounded. 
  
  The other direction is more challenging and involves proving that the every direction with a finite maximizer over $\Safe(\Ap,\Am)$ is in $\Dir(\Ap,\Am)$. We would like to use linear programming (LP) duality to prove that every direction is in the cone $\Dir(\Ap,\Am)$. However, it is not clear that $\Safe$ is a polyhedron to apply LP duality i.e., described by a finite set of linear inequalities. Note that from \eqref{eq:cone:desc}, we see that $\Dir(\Ap, \Am)$ is described by a finite number of constraints. Hence by Weyl's theorem on polyhedral cones (see \cite{Schrijverbook}), $\Dir(\Ap, \Am)$ is also a finitely generated cone. However this does not suffice since the constraint for each direction is of the form $d \cdot x \le \max_d(\Am)$. \footnote{For example even if $d = v_1+v_2$, 
  the RHS of the constraint $\max_d(\Am)$ could be smaller than $\max_{v_1}(\Am)+\max_{v_2}(\Am)$; hence some of the constraints that define $\Safe$ are not necessarily implied by constraints on just the generators of the cone $\Dir(\Ap,\Am)$.}

  We first show that $\Safe$ can indeed be described by a finite number of linear inequalities, and then use LP duality to complete the argument. 
  Let $\ell = |\Am|$ and $\Am=\{a_1, a_2, \dots, a_\ell \}$. We define convex sets $\Dir_i$ and $\Safe_i$ (here we suppress the arguments $\Ap, \Am$ for easier notation) as follows :
  \begin{align}
      \forall i \in [\ell], ~ \Dir_i & \coloneqq \{ d \in \Dir(\Ap, \Am) : d \cdot a_i=\max_{d}(\Am) \}.\\
      \Safe_i &\coloneqq \Big\{ x : \forall d \in \Dir_i, ~ d \cdot x \le d \cdot a_i \Big\}  \\
      \text{Then, } \Safe(\Ap, \Am) &= \bigcap_{i \in [\ell]} \Safe_i. \label{eq:safes}
  \end{align}
  
We now show that each of the convex sets $\Safe_i$ (and hence $\Safe$) is polyhedral i.e., described by a finite set of constraints. 
For each $i \in [\ell]$, $\Dir_i$ is also a convex cone that is finitely generated. This is because $\Dir_i$ is described exactly by the finite set of linear constraints as $\Dir_i = \{ d : \forall j \in [\ell], ~ d \cdot (a_j - a_i) \le 0 \}$. 
  Hence by the Weyl theorem for polyhedral cone duality (see \cite{Schrijverbook}), there exists finite $r_i \in \mathbb{N}$ such that the set of vectors $v_{i1}, \dots, v_{ir_i}$ such that $\Dir_i=\cone(v_{i1}, \dots, v_{ir_i})$. Now we see that 
  \begin{equation} \label{eq:safei}
  \Safe_i = \{ x : \forall j \in [r_i], ~ v_{ij} \cdot x \le v_{ij} \cdot a_i\}. 
  \end{equation}
  The subset inclusion in \eqref{eq:safei} is obvious. The other direction just follows because $v_{i1}, \dots, v_{ir_i}$ generate the cone. This shows that for each $i \in [\ell]$, $\Safe_i$  is polyhedral i.e., described by a finite set of linear constraints. Hence from \eqref{eq:safes} we have 
  \begin{equation}\label{eq:safelp}
  \Safe(\Ap, \Am) = \{x : \forall i \in [\ell], \forall j \in [r_i], v_{ij} \cdot x \le b_{ij} \},
  \end{equation}
  which is described by a finite number of constraints $r_{\text{tot}} \coloneqq \sum_{i =1}^\ell r_i$. 
  
Finally, we now use linear programming duality to show that $d$ has a finite maximum over $\Safe$ if and only if $d \in \Dir(\Ap, \Am)$. Consider the linear program (LP) given by $\max_{x \in \R^n} d \cdot x$ such that $x$ satisfies the constraints in \eqref{eq:safelp}. By LP duality, this LP has a finite maximum (i.e., bounded) if and only if its dual LP is feasible i.e., there exists a non-negative vector $y \ge 0$ in $r_{\text{tot}}$ dimensions with such that 
$$ \sum_{i=1}^\ell \sum_{j=1}^{r_i} y_{ij} v_{ij} = c.$$
In other words, if $\max_d(\Safe(\Ap,\Am))$ is finite (bounded), then $d \in \cone(\{v_{ij}: i \in [\ell], j \in [r_i]\}) \subset \Dir(\Ap, \Am)$.   
\end{proof}

}

Now we show that the vertices of $\Safe$ are equal to the critical points.

\begin{lemma}
  \label{l:connecting}
  For any linearly separable set $\Sp \sqcup \Sm = S$, the critical points are the vertices of the safe points, i.e., $$\CP(\Sp) = \Verts(\Safe(\Sp,\Sm)).$$
\end{lemma}

\begin{proof}
  Let $S' = \Safe(\Sp,\Sm)$ and $V' = \Verts(S')$. The proof follows from the following two statements that we establish:
  \begin{align}
     V' &\subseteq \Sm \label{eq:verts:1}\\    
      V' & \subseteq \CP(\Sp). \label{eq:verts:2}\\
      V' & \supseteq \CP(\Sp). \label{eq:verts:3}
  \end{align}
  
Consider any vertex $x' \in S'$, and let $d' \in \R^n$ be the direction that it uniquely maximizes within $S'$. 
\nsubtracted{
Since $\Safe$ is defined as the intersection of halfspaces (constraints on the extent in directions $\Dir$), duality implies that these are the only directions with finite maximizers:
$$ \max_{x \in S'} d' \cdot x \text{ is bounded } \Leftrightarrow d' \in \cone(\Dir(\Sp, \Sm)) = \Dir(\Sp, \Sm)    .
$$
Note that the last equality in the above equation follows since $\Dir(\Sp, \Sm)$ is closed under scaling and convex combinations. 
}
\nadded{From Lemma~\ref{l:finitemax} we have that $d' \in \Dir(\Sp, \Sm)$ since $d'$ has a finite maximizer in $\Safe(\Sp, \Sm)$. }
From the definition of $\Safe(\Sp, \Sm)$, for every direction $d \in \Dir(\Sp, \Sm)$ (and in particular $d'$), there exists an element of $\Sm$ that achieves $\max_{x \in S'} d \cdot x$. Hence $x' \in \Sm$ since $x'$ is the unique maximizer in $S'$. This establishes \eqref{eq:verts:1}. 

We now show $V' \subseteq \CP(\Sp)$. As before let $x' \in V'$ and $d'$ be a direction that it uniquely maximizes within $S' \supseteq \Sm$. Hence there is a linear classifier consistent with $\Sp \cup \{x'\}$ labeled positive, and $\Sm \setminus \{x'\}$ labeled negative. Hence $x' \in \CP(\Sp)$ as required for \eqref{eq:verts:2}. 

Finally to show \eqref{eq:verts:3}, suppose $x' \in \CP(\Sp)$. By definition, there is a linear classifier separating $\Sp \cup \{ x' \}$ (as positives) and $\Sm \setminus \{x'\}$ (as negatives). Moreover, $\Sp$ and $\Sm$ are also separable (in particular it labels $x' \in \Sm$ as a negative example). Hence by convexity, there exists a direction $d'$ such that $x'$ is the unique maximizer among the $\Sm$, and $d' \cdot x' < \min_d(\Sp)$. Hence $d' \in \Dir(\Sp, \Sm)$ and from \eqref{eq:verts:1}, we have that $x'$ is a unique maximizer in $S'$ of $d'$. Hence \eqref{eq:verts:3}. This concludes the proof.
\end{proof}


\begin{proof}[Proof of \Cref{thm:verts}]
  By \Cref{l:connecting}, the critical points of $\CP(\Sp)$ are equal
  to the vertices of $\Safe(\Sp,\Sm)$.  Plugging this equivalence into
  \Cref{c:leak-verts}, we have $\Leak(\Sp,\CP(\Sp)) = \Sp \cup
  \CP(\Sp)$.
\end{proof}

\subsection{Truthful Protocols}
\label{s:truth}

In protocols for e-discovery Alice desires to (a) hide positive data
points and (b) reduce the disclosure of negative data points.  In a
correct protocol, all positive data points are revealed, thus, Alice
faces only the problem of reducing the discosure of negative data
points.  Following the standard framework from mechanism design in
economics and computer science, we define the protocol properties of
direct and truthful and provide a revelation principle.  For this
discussion we view Alice's interaction in the protocol.  Note that all
correct protocols disclose all of the positive points $\Sp$; thus,
minimizing the number of points disclosed in a correct protocol is
equivalent to minimizing the number of negative points disclosed.

\begin{definition}
  Given a known set of points $S$, a direct protocol
  $\proto : 2^S \times 2^S \to 2^S$ maps the sets of alleged positives
  (of Alice) and true positives (as can be verified by Bob) to a set
  of disclosed points (to Bob).
\end{definition}

\begin{definition}
  A direct protocol is {\em truthful} if for all $\Sp
  \sqcup \Sm = S$, Alice's optimal strategy is to truthfully report
  $\Sp$.
  \end{definition}

\begin{proposition}[Revelation Principle]
  \label{p:rev}
  For any protocol $\proto$ and optimal strategy $\strat$ of Alice
  mapping positive points to messages in the protocol (which minimizes
  the total number of data points disclosed), there is a truthful
  and direct protocol $\protor$ with the same outcome under
  truthtelling (as under protocol $\proto$ with strategy $\strat$).
\end{proposition}

\begin{proof}
  Define the revelation protocol as $\protor(\Ap,\Sp)$ as follows:
  \begin{enumerate}
  \item Simulate
    strategy $\strat(\Ap)$ in $\proto$ assuming $\Ap$ are the true
    positives.
  \item Given the transcript of this simulation, attempt the same interaction as $\strat(\Ap)$ with the real true positives
    $\Sp$.
  \item If the behavior of $\proto$ with true positives $\Sp$ is ever
    deviates from the simulated transcript, then reveal
    the full dataset $S$ to Bob.  (Otherwise, the outcome is
    identical to $\proto(\strat(\Ap),\Ap)$.)
  \end{enumerate}

  We now argue that the optimal strategy in $\protor$ is to report
  $\Ap = \Sp$.  Suppose some $\Ap \neq \Sp$ gives a strictly better
  outcome.  Note: it must be that the outcomes of
  $\proto(\strat(\Ap),\Ap)$ and $\proto(\strat(\Ap),\Sp)$ are the
  same, otherwise, the difference would be detected and the full set
  $S$ would be disclosed to Bob.  In this case, however, with true
  positives $\Sp$ following $\strat(\Ap)$ rather than $\strat(\Sp)$ in
  $\proto$ gives a strictly better outcome, which contradicts the
  optimality of $\strat$ for $\proto$.
  \end{proof}

\subsection{Minimal Protocol}

In this section we prove that every correct protocol discloses the
critical points $\CP(\Sp)$ on dataset $\Sp \sqcup \Sm = S$; thus, the
critical points protocol is optimal.

\begin{theorem}
  \label{prop:minimal}
  Every correct protocol $\proto$ on dataset $\Sp \sqcup \Sm = S$
  discloses a set of points that contains $\CP(\Sp)$.
\end{theorem}

\begin{proof}
  By \Cref{p:rev}, it is without loss to assume $\proto$ is truthful.
  By the definition of truthful protocols, Alice cannot have fewer points
  disclosed by reporting non-truthfully.  Suppose for a contradiction
  that a point $x^* \in \CP(\Sp)$ is not disclosed in $\Sp \sqcup \Am =
  \proto(\Sp,\Sp)$, i.e., $x^* \in \Sm$ but $x^* \not\in \Am$.

  Recall $\CP(\Sp) = \{x \in S : \Leak(\Sp,S \setminus \Sp \setminus
  \{x\})\}$ is the points that are each labeled as positive by some
  consistent classifier with respect to positives $\Sp$ and negatives
  $S \setminus \Sp \setminus \{x\}$.  Monotonicity of $\Leak$ implies
  that $x^*$ is in $\Leak(\Sp,\Am)$ as $\Am \subseteq S \setminus \Sp
  \setminus \{x^*\}$.  By the definition of $\Leak$ there is a
  consistent classifier that labels $\Sp \cup \{x^*\}$ as positive and
  $\Am$ as negative.

    Since $x^*$ is in
  $\Leak(\Sp,\Am)$, there exists a separating hyperplane for $\Sp \cup
  X$ (with $X \ni x^*$) and the remaining points (which contains
  $\Am$). Thus, \Cref{l:const-not-disclosed} (below) can be applied where $\Sp
  \cup X$ is separable but $x^* \in X$ is not disclosed on $\proto(\Sp,\Sp
  \cup X)$, a contradiction to the correctness of $\proto$ as $x^* \in
  X$ is considered a positive point in the execution of $\proto(\Sp,\Sp
  \cup X)$
\end{proof}

\begin{lemma}
  \label{l:const-not-disclosed}
  In any direct protocol $\proto$, if $X \sqcup \Sp$ is separable and
  $X \cap \proto(\Ap,\Sp) = \emptyset$ (i.e., $X$ is not disclosed by
  $\proto$ with $X$ are negative) then $\proto(\Ap,\Sp) =
  \proto(\Ap,\Sp \cup X)$ (i.e., $\proto$ on $\Ap$ discloses the same points when points $X$ are all positive or all negative).
\end{lemma}

\begin{proof}
  Separability of $X\cup \Sp$ implies that $\proto(\Ap,\Sp \cup X)$ is
  well defined.  By definition a protocol is only a function of its
  input, in this case, $\Ap$ and the points that it discloses.  Since
  $X$ is not disclosed in $\proto(\Ap,\Sp)$ then $\proto(\Ap,\Sp \cup
  X)$ has the same result, and $X$ is not disclosed by it as well.
\end{proof}

\section{Machine Learning Guarantees} \label{sec:mlframework}

\newcommand{\algoclassifier}{\text{Alg}_{\cH}}
\newcommand{\htrain}{\hat{h}}

In this section we consider the machine learning framework related to multi-party e-discovery (MPeD). We show how the protocol defined in the previous section (for linear classification), when instantiated with {\em any} training algorithm that satisfies a natural property, that we call the {\em independence of irrelevant alternatives (IIA)}, achieves the same learning guarantees as single-party e-discovery (SPeD). We provide a reduction from MPeD to SPeD; this shows that our protocol suffers no loss in generalization or sample complexity compared to the standard single-party setting. Finally, we instantiate this reduction using the classic support vector machine (SVM) algorithm, by showing that it satisfies the IIA property.

\subsection{Machine learning framework}

We start by recalling the single party e-discovery (SPeD) framework which corresponds to a standard machine learning pipeline, involving a training algorithm run on the training set $S$ to find a good classifier $\htrain$, and then applying this classifier on the entire dataset $U$.  $\algoclassifier$ will denote a learning algorithm for the hypothesis class $\cH$ that takes in labeled samples as input, and outputs a hypothesis in $\cH$ consistent with the labeled samples.



\renewcommand{\algorithmicrequire}{\textbf{Input:}}
\renewcommand{\algorithmicensure}{\textbf{Output:}}

\begin{algorithm}
\caption{ML framework for Single-Party e-Discovery (SPeD)}\label{alg:ml:sped}
\begin{algorithmic}[1]
\Require Unlabeled dataset $U$ and the hypothesis class $\cH$. 
\Ensure Data points with positive labels $U_+\subseteq U$.
\State Sample the training data set $S\subseteq U$ (likely with $|S|\ll |U|$ ).
\State (Hand-)Label $S$ to get $S_+$ and $S_-$.
\State Use $\algoclassifier$ to learn classifier $\htrain\in \cH$ on the labeled data with positives $S_+$ and negatives $S_-$.
\State Apply classifier $\htrain$ on $U$ to get $U_+$
\end{algorithmic}
\end{algorithm}



In the Multi-Party e-Discovery (MPeD) framework, there are three parties Alice, Bob and {\TTP} that perform different functions. Alice first sends the entire dataset $U$ to {\TTP}. {\TTP} generates the training samples $S \subseteq U$ and sends it to Alice. Alice, {\TTP} and Bob engage in a protocol as in Section~\ref{sec:ccp}. 
At the end of the protocol, Bob is given a set of positive examples $S'_+$ (which is hopefully $S_+$) and some other negative samples $S'_-$ (which is hopefully much smaller than $S_-$), which he then uses to train 
a classifier $\htrain$ that is used to classify the entire dataset $U$. 


 \renewcommand{\algorithmicrequire}{\textbf{Input:}}
 \renewcommand{\algorithmicensure}{\textbf{Goal:}}

\begin{algorithm}
\caption{ML framework for Multi-Party e-Discovery (MPeD)}\label{alg:ml:mped}
\begin{algorithmic}[1]
\Require Alice has unlabeled data points $U$. The hypothesis class is $\cH$ is known publicly.  
\Ensure Bob receives data points with positive labels $U_+\subseteq U$.
\State Alice sends entire dataset $U$ to {\TTP}. 
\State \cTTP samples the training data set $S \subseteq U$ (likely with $|S| \ll |U|$).
\State {\TTP} sends $S$ to Alice. 
\State Alice, {\TTP}, Bob participate in the {\em critical points protocol} (\Cref{alg:ccp}) of Section~\ref{sec:ccp}. At the end of it, Bob receives labeled samples $S'_+$ and $S'_-$.   
\State Bob uses $\algoclassifier$ to learn a classifier $\htrain\in \cH$ consistent with the labeled data $S'_+$ and $S'_-$, and sends $\htrain$ to {\TTP}. 
\State {\TTP} checks the consistency of $\htrain$ with $S'_+, S'_-$ and applies $\htrain$ on $U$ to get $U_+$ and sends it to Bob.
\end{algorithmic}
\end{algorithm}

We want the classifier that is output in multi-party ML framework to be as good as the classifier in the single-party setting, irrespective of Alice's actions; ideally, it also  maintains the same statistical properties (e.g., sample complexity) as the single-party setting. However, the choice of the learning algorithm is important, since the algorithm is trained on a different set of labeled samples ($(S'_+, S'_-)$ as opposed to $(S_+, S_-)$). The following property of the learning algorithm will play a crucial role.   
Recall that $\cH(\Sp, \Sm)$ denotes the set of hypothesis in $\cH$ consistent with the labeled data given by positives $\Sp$ and negatives $\Sm$; also for a linearly separable data set $\Sp \sqcup \Sm$, the critical points are denoted by $\CP(\Sp)$.

\begin{definition} (IIA property) \label{def:IIA}
	A learning algorithm $\algoclassifier$ is said to satisfy \emph{independence of irrelevant alternatives} (IIA) if for 
any $S_+$ and $S_- \supseteq \CP(\Sp)$ that is separable and for any $S'_- \supseteq \CP(\Sp)$, we have that $\algoclassifier(S_+,S_-)=\algoclassifier(S_+,S'_-)$.  
\end{definition}
One can also define a potentially stronger notion of IIA where  $\cH(S_+,S_-)=\cH(T_+,T_-)\implies \algoclassifier(S_+,S_-)=\algoclassifier(T_+,T_-)$, but the above weaker notion suffices for our purposes. 


We focus on the setting where the data set $U$ is linearly separable i.e., $\cH$ is the set of linear classifiers, and there is an $h^* \in \cH$ that is consistent with the true labels of $U$. 
We now show that we can use our protocol from Section~\ref{sec:ccp} in Step 4 of the above framework, along with any algorithm that satisfies the IIA property to achieve the same statistical guarantees as the single-party ML setting.  Note that the IIA property pertains only to the learning algorithm that is employed by Bob in Step 6 of the \Cref{alg:ml:mped}. 

\begin{theorem}\label{prop:IIA}
	Suppose the data set $U$ is linearly separable and the learning algorithm $\algoclassifier$ satisfies the IIA property. Then with the same sampling procedure (to produce $S$), the outputs of \Cref{alg:ml:sped} (SPeD) and \Cref{alg:ml:mped} (MPeD) are identical.
\end{theorem}
We remark that there can be randomness in the sampling procedure and potential random choices in the learning algorithm $\algoclassifier$; so $U_+$ is a random set. The guarantee of \Cref{prop:IIA} is that the distributions of $U_+$ are the same. Alternatively, fixing the random choices in the sampling, and in the algorithm $\algoclassifier$, the set $U_+$ is the same. 
\begin{proof}[Proof of \Cref{prop:IIA}]
	The proof follows easily by combining the guarantees of 
	\Cref{thm:min} and the IIA property of $\algoclassifier$.  
	First, from \Cref{thm:min}, we know that irrespective of the actions of Alice, Bob receives $S_+$ and $C^*(S_+)$, where $C^*(S_+)$ denotes the critical points. 
	Hence, from the IIA property of $\algoclassifier$, the classifier $\htrain$ that is produced is identical (for the same random choices of the algorithm $\algoclassifier$. Hence $U_+$ is identical in both cases.    
\end{proof}
 
For a given linearly separable dataset there may be several potential linear classifiers consistent with it (and so too for $(S_+, S_-)$). Furthermore, not all algorithms for linear classification may satisfy the IIA property (e.g., the popular Perceptron algorithm does not satisfy IIA). However, the well-known SVM algorithm that finds the maximum margin classifier for the given linearly-separable dataset satisfies the IIA property (see \Cref{lem:svm:iia} in the next section). Hence we can instantiate \Cref{prop:IIA} for linear classifiers by using the SVM algorithm as follows. (The proof just follows by combining \Cref{prop:IIA} and \Cref{lem:svm:iia}.)
 
 \begin{theorem}\label{corr:svm:ml}
	Suppose the data set $U$ is linearly separable and the learning algorithm $\algoclassifier$ is the SVM algorithm given in Section~\ref{sec:svm}. Then with the same sampling procedure (to produce $S$), the outputs of \Cref{alg:ml:sped} (SPeD) and \Cref{alg:ml:mped} (MPeD) are identical.
\end{theorem}

 This theorem shows that the multi-party e-discovery protocol given in \Cref{alg:ml:mped} incurs no loss compared to the single-party setting (\Cref{alg:ml:sped}) in terms of properties of the output classifier $\htrain$. In particular, any statistical property (like test error or generalization guarantee) of the classifier $\htrain$ transfer over to the multi-party setting with no loss in the statistical efficiency. 
See Section~\ref{sec:kernel} for the extension to kernel classifiers.



\newcommand{\wdagger}{w^{\dagger}}
\newcommand{\bdagger}{b^{\dagger}}

\subsection{Support Vector Machines (SVM) and Properties} \label{sec:svm}

We now describe the support vector machine (SVM) algorithm which is used for learning linear classifiers for a given set of labeled samples in high-dimensional spaces. We will also prove the IIA property and see some facts about the SVM algorithm that will be useful in the next section. 

The setting is as follows. We are given a set of labeled samples $T=\{(x_1, y_1), (x_2, y_2), \dots, (x_m, y_m) \} \subset \mathbb{R}^n \times \{\pm 1\}$. The goal is to find a linear classifier $(w,b) \in \mathbb{R}^n \times \R$ such that $\forall i \in [m], ~ y_i ({w \cdot x_i} + b) >0$ if such a classifier exists i.e., it is linearly separable.

The Hard-SVM algorithm finds the linear classifier that separates the positive and negative samples with the largest possible margin (if the data is linearly separable). For a classifier $(w,b)$ with $\| w \|_2 =1$, the margin of a sample $(x,y)$ is the distance between a point $x$ and the hyperplane $(w,b)$ and is given by $\max \big\{ y({w \cdot x}+b) ,0 \big\}$. Note that a linear classifier does not change by scaling. 
The following claim shows that the \textsc{Hard-SVM} problem of finding maximum margin linear classifier can be reformulated in either of the following two ways. 
\begin{fact}\citep
[See e.g., chapter 15 of][]{ShaiBSSbook}\label{fact:svm1}
Given a set of linearly separable labeled samples \\
$\{(x_1, y_1), (x_2, y_2), \dots, (x_m, y_m) \} \subset \mathbb{R}^n \times \{\pm 1\}$, consider the following optimization problems:
\begin{align}\label{eq:svm1}
    w^*, b^*&=\argmax_{(w,b): \|w \|_2 =1} \min_{i \in [m]} \lvert w \cdot x_i+b \rvert  \\
 \text{ s.t. } &\forall i \in [m], ~~ y_i \big( w \cdot x_i+b \big)>0 \nonumber \\
\label{eq:svm2}
    \wdagger, b^{\dagger}&=\argmin_{(w,b)} \| w \|_2^2  \\ \text{ s.t. } & \forall i \in [m],~~ y_i \big( {w \cdot x_i}+b \big)  \ge 1. \nonumber
 \end{align}
The optimization problems \eqref{eq:svm1} and \eqref{eq:svm2} are essentially equivalent, with the optimizers related as $w^* = \frac{\wdagger}{\|\wdagger\|_2}, b^* = \frac{\bdagger}{\|\wdagger\|_2}$. Moreover \eqref{eq:svm2} is a convex program that can be solved in polynomial time.  

\end{fact}
While \eqref{eq:svm1} more directly captures the maximum margin formulation, \eqref{eq:svm2} more clearly illustrates why it is a convex program that can be solved in polynomial time. When the data is not linearly separable, the linear constraints that define the convex program in \eqref{eq:svm2} just become infeasible. 

We have the following characterization that the solution of the \textsc{Hard-SVM} problem can be expressed as a linear combination of points which are all at the minimum distance (of exactly $1/\|\wdagger\|_2$) from the separating hyperplane (these points are called the support vectors). Note that the optimal solution of \eqref{eq:svm2} when it exists, is always unique (the objective is strongly convex). 
\begin{fact}\citep[Theorem 15.8 in][]{ShaiBSSbook}]\label{fact:svm2}
Let $\wdagger, \bdagger$ denote the optimal solution of \eqref{eq:svm2}, and let $I=\{i \in [m]: y_i (\wdagger \cdot x_i) =1\}$. Then there exists coefficients $\alpha_1, \dots, \alpha_m \in \mathbb{R}$ such that $\wdagger=\sum_{i \in I} \alpha_i x_i$.  
\end{fact}

 The above fact can be used to reformulate the objective \eqref{eq:svm2} in terms of the unknowns $\alpha_1, \dots, \alpha_m$ and the inner products between the points $\{ x_i \cdot x_j: i, j \in [m] \}$ (as opposed to the points $\{x_i: i \in [m] \}$ themselves). 

\begin{lemma}\label{lem:svm:iia}(SVM satisfies IIA property)
 The \textsc{Hard-SVM} procedure given by \eqref{eq:svm2} is IIA (as per Definition~\ref{def:IIA}) i.e., for any $S_+$ and $S'_- \supseteq C^*(S_+)$ (so that $\mathcal{H}(S_+, S'_-)=\mathcal{H}(S_+, S_-)$), the solution $w',b'$ on the labeled examples given by $(S_+, S'_-)$ is identical to the solution $w,b$ on the labeled examples $(S_+, S_-)$.    
\end{lemma}
We remark that SVM satisfies a stronger notion of IIA, where the positives can also be a subset $S'_+ \subseteq S_+$ such that the set of consistent hypothesis remains the same. However the above version suffices for our purposes.  
 \begin{proof}


Consider the SVM solution $w,b$ (for \eqref{eq:svm2}) on labeled data $S_+, S_-$. This solution is unique: if there are two minimizing solutions $(w,b), (w',b')$ then the solution $(\tfrac{1}{2}(w+w'), \tfrac{1}{2}(b+b'))$ also satisfies all the constraints of \eqref{eq:svm2} but attains a smaller objective value due to strong convexity. 

Consider the classifier $w',b'$ that attains the optimum margin for the dataset $S_+, S'_-$. From \Cref{c:same} (and \Cref{thm:verts}) classifier $(w',b') \in \mathcal{H}(S_+, S_-)$. Moreover the (minimum) margin over $S_+, S_-$ is attained by a point in $S_+ \cup S'_-$ i.e., $\min_{x \in S_+ \cup S'_-} |w'\cdot x + b'| = \min_{x \in S_+ \cup S_-} |w'\cdot x + b'| \eqqcolon \tau$.\footnote{Note that not all of the support vectors from Fact~\ref{fact:svm2} need to be in $S_+ \cup S'_-$. } Suppose not. There exists a point $x' \notin S_+ \cup S'_-$ with label $y'$ (say $y'=-1$) such that $y'(w'\cdot x' + b') = \tau$, but for all labeled examples $(x,y)$ given by $S_+ \cup S'_-$, $y(w'\cdot x+b') >\tau$. Then $x' \in S_+ \cup C^*(S_+) \subseteq  S_+ \cup S'_-$ (if $y=-1$, $x' \in C^*(S_+)$) which gives a contradiction. 

Hence $w',b'$ also achieves the same margin on $S_+, S_-$ (as it does on $S_+, C^*(S_+)$). This implies that $w',b'$ is also a optimum margin classifier on $S_+, S_-$ (since it had a larger margin than $w,b$ on $A^*_+, A^*_-$). Since the solution to \eqref{eq:svm2} is unique, we conclude that $(w,b) = (w',b')$. 
\end{proof}

\section{Extensions to Kernel Classifiers} \label{sec:kernel}


Our results for linear classifiers in the previous section naturally extend to kernel-based classifiers like kernel support vector machines. Kernel-based classifiers can be much more expressive than linear classifiers as they embed the input space into a feature space that is high-dimensional (potentially infinite dimensional), where the data is potentially linearly separable. Popular kernels include the polynomial kernel, that can capture any polynomial threshold function (a classifier of the form $\text{sign}(p(x))$ where $p(x)$ is any polynomial of $x$), radial basis kernels (e.g., Gaussian kernels) etc.

We start by recalling some notation and facts about kernels. In what follows, $\psi: \mathcal{X} \to \mathcal{V}$ embeds points in the input space $\mathcal{X}$ into a Hilbert space $\mathcal{V}$. The kernel function is given by the inner product $K(x,x') = \iprod{\psi(x), \psi(x')}_{\mathcal{V}}$, and we assume that this can be computed in polynomial time given $x,x'$. Given a set of points $x_1, \dots, x_m \in \mathcal{X}$, the $m \times m$ matrix formed with the $(i,j)$th entry $K(x_i, x_j)$ is positive semi-definite. The following standard theorem states that one can find the maximum margin classifier in the feature space $\psi(\mathcal{X})$ for a given set of $m$ samples, by solving a simpler convex optimization problem  over $m$ dimensions.  

\begin{theorem}[see e.g., Chapter 16 of ~\citet{ShaiBSSbook}] \label{thm:kernelSVM} For a given set of samples $(x_1,y_1), \dots, (x_m,y_m) \in \mathcal{X} \times \{ \pm 1\}$, consider the  Hard-SVM problem (or the maximum margin classifier problem): 
\begin{equation} \label{eq:kernelSVM}
\min_{w \in \mathcal{V}, b \in \R} \|w \|_2^2, \text{ s.t. } ~\forall i \in [m],~~ y_i \Big( {w \cdot \psi(x_i)} +b \Big) \ge 1.    
\end{equation}
%
An optimal solution $(w^*, b^*)$ to \eqref{eq:kernelSVM} can be obtained in polynomial time by finding an optimal solution to the following convex optimization problem and setting $w^*=\sum_{i=1}^m \alpha^*_i \psi(x_i)$:
\begin{align} \label{eq:kernelSVM2}
(\alpha^*, b^*)&=\argmin_{\alpha \in \R^m, b \in \R} \sum_{i,j \in [m]} K(x_i, x_j) \alpha_i \alpha_j, \nonumber\\
\text{ s.t. } ~& \forall i \in [m],~~ y_i \Big( \sum_{j \in [m]} K(x_i,x_j) \alpha_j  +b \Big) \ge 1.
\end{align}
Moreover the corresponding linear classifier is given by $h(x) = \text{sign}\big(\sum_{j \in [m]} \alpha_j K(x_j, x) + b \big)$. 
\end{theorem}

Our main observation is that implementing the protocol in the previous section only involves solving a set of linear classification problems. These include the computation of the set $C^*(A_+)$ by {\TTP} (one problem for each point in $S_-$), and the computation of the set $\Leak(A_+, A_-)$ by {\TTP} (one problem for each point in $S$). These linear classification problems are solved using the Hard-SVM problem (maximum margin linear classifier), which also satisfies the IIA property. We can carry out the same arguments as in the previous sections on the points $\{(\psi(x_i), y_i): i \in [m]\}$ by only accessing inner products $\{ \psi(x_i) \cdot \psi(x_j): i,j \in [m]\}$ (through the kernel function). Hence from Theorem~\ref{thm:kernelSVM} we immediately obtain the generalization of our guarantees to kernel-based classifiers.

\begin{theorem}
Suppose a kernel function $K: \mathcal{X} \times \mathcal{X} \to \R$ is specified as above and efficiently computable (and known to Alice, Bob and {\TTP}). Given a set of points $S=(S_+, S_-)$ that is realizable (w.r.t. to linear classifiers over the space $\psi(\mathcal{X})$), there is a protocol (the critical points protocol implemented using kernel SVM) that is correct, minimal, computationally efficient and truthful. 
Moreover, if the data set $U$ is realizable (w.r.t. linear classifiers over the space $\psi(\mathcal{X})$), Bob uses the kernel SVM algorithm in Step 6 of \Cref{alg:ml:mped} (MPeD), then the outputs of \Cref{alg:ml:sped} (SPeD) and \Cref{alg:ml:mped} (MPeD) are identical assuming the same sampling procedure (for generating $S$ from $U$).
\end{theorem}

\section{Discussion and Open Questions}


In this paper, we considered a multi-party classification problem that is motivated by e-discovery. 
We designed a protocol (critical points protocol) in the realizable setting of linear classifiers and kernel SVMs, that is correct, minimal, computationally efficient and truthful. Moreover this protocol fits into a machine learning framework along with any classifier that satisfies the a natural property called IIA; we provide a reduction to the standard single-party setting, thereby demonstrating that there is no loss in statistical efficiency.

The most natural direction for future research is to generalize to the non-realizable setting, where there is no perfect classifier from the hypothesis class. Here we may need to relax our requirement of computational efficiency, as the problem of learning a linear classifier in the non-realizable setting (also called agnostic learning) is known to be computationally intractable even in the single-party setting~\cite{ShaiBSSbook}. However one may be able to obtain efficient protocols assuming access to an efficient (single-party) learning algorithm. 







\subsection*{Acknowledgements}
  This work began during the IDEAL Special Quarter on Data
  Science and Law organized by Jason Hartline and Dan Linna.  Many
  thanks to Dan Linna for legal context and feedback on the project.
  The work was supported in part by NSF award CCF-1934931.
\bibliographystyle{apalike}
\bibliography{ediscovery}

\begin{thebibliography}{}

\bibitem[Bertsimas and Tsitsiklis, 1997]{bertsimas-LPbook}
Bertsimas, D. and Tsitsiklis, J. (1997).
\newblock {\em Introduction to linear optimization}.
\newblock Athena Scientific.

\bibitem[Cormack and Grossman, 2014]{CG-14}
Cormack, G.~V. and Grossman, M.~R. (2014).
\newblock Evaluation of machine-learning protocols for technology-assisted
  review in electronic discovery.
\newblock In {\em Proceedings of the 37th international ACM SIGIR conference on
  Research \& development in information retrieval}, pages 153--162.

\bibitem[Dwork et~al., 2012]{DHPR-12}
Dwork, C., Hardt, M., Pitassi, T., Reingold, O., and Zemel, R. (2012).
\newblock Fairness through awareness.
\newblock In {\em Proceedings of the 3rd innovations in theoretical computer
  science conference}, pages 214--226.

\bibitem[Gelbach and Kobayashi, 2015]{GK-15}
Gelbach, J.~B. and Kobayashi, B.~H. (2015).
\newblock The law and economics of proportionality in discovery.
\newblock {\em Ga. L. Rev.}, 50:1093.

\bibitem[Goldreich et~al., 1987]{GMW-87}
Goldreich, O., Micali, S., and Wigderson, A. (1987).
\newblock How to solve any protocol problem.
\newblock In {\em Proc. of STOC}.

\bibitem[Goldwasser et~al., 2021]{GRSY-21}
Goldwasser, S., Rothblum, G.~N., Shafer, J., and Yehudayoff, A. (2021).
\newblock Interactive proofs for verifying machine learning.
\newblock In {\em 12th Innovations in Theoretical Computer Science Conference}.
  Schloss Dagstuhl-Leibniz-Zentrum f{\"u}r Informatik.

\bibitem[Grossman and Cormack, 2010]{GC-10}
Grossman, M.~R. and Cormack, G.~V. (2010).
\newblock Technology-assisted review in e-discovery can be more effective and
  more efficient than exhaustive manual review.
\newblock {\em Rich. JL \& Tech.}, 17:1.

\bibitem[Jacot et~al., 2018]{jacotetal}
Jacot, A., Gabriel, F., and Hongler, C. (2018).
\newblock Neural tangent kernel: Convergence and generalization in neural
  networks.
\newblock In {\em Proceedings of the 32nd International Conference on Neural
  Information Processing Systems}, page 8580–8589, Red Hook, NY, USA. Curran
  Associates Inc.

\bibitem[Scholkopf and Smola, 2001]{kernelbook}
Scholkopf, B. and Smola, A.~J. (2001).
\newblock {\em Learning with Kernels: Support Vector Machines, Regularization,
  Optimization, and Beyond}.
\newblock MIT Press, Cambridge, MA, USA.

\bibitem[Schrijver, 1999]{Schrijverbook}
Schrijver, A. (1999).
\newblock Theory of linear and integer programming.
\newblock In {\em Wiley-Interscience series in discrete mathematics and
  optimization}.

\bibitem[Shalev-Shwartz and Ben-David, 2014]{ShaiBSSbook}
Shalev-Shwartz, S. and Ben-David, S. (2014).
\newblock {\em Understanding Machine Learning - From Theory to Algorithms.}
\newblock Cambridge University Press.

\bibitem[Yao, 1986]{yao-86}
Yao, A. C.-C. (1986).
\newblock How to generate and exchange secrets.
\newblock In {\em 27th Annual Symposium on Foundations of Computer Science},
  pages 162--167. IEEE.

\end{thebibliography}

\appendix
\end{document}